%% file: revision1.tex
\documentclass[pdflatex,sn-mathphys-num]{sn-jnl}

\usepackage{graphicx}%
\usepackage{multirow}%
\usepackage{amsmath,amssymb,amsfonts}%
\usepackage{amsthm}%
\usepackage{mathrsfs}%
\usepackage[title]{appendix}%
\usepackage{xcolor}%
\usepackage{textcomp}%
\usepackage{manyfoot}%
\usepackage{booktabs}%
\usepackage{algorithm}%
\usepackage{algorithmicx}%
\usepackage{algpseudocode}%
\usepackage{listings}%

\theoremstyle{thmstyleone}%
\newtheorem{theorem}{Theorem}
\newtheorem{lemma}[theorem]{Lemma}
\newtheorem{corollary}[theorem]{Corollary}

\theoremstyle{thmstyletwo}%

\theoremstyle{thmstylethree}%
\newtheorem{definition}{Definition}%

\DeclareMathOperator{\oracle}{oracle}
\DeclareMathOperator{\dom}{dom}
\DeclareMathOperator{\reg}{region}
\newcommand{\seq}{$s$}

\usepackage{yfonts}

\begin{document}

\title[An Efficient Algorithm for Exploring RNA Branching Conformations under the NNTM]{ An Efficient Algorithm for Exploring RNA Branching Conformations under the Nearest-Neighbor Thermodynamic Model}


\author[1]{\fnm{Svetlana} \sur{Poznanovi\'c}}\email{spoznan@clemson.edu}

\author[2]{\fnm{Owen} \sur{Cardwell}}\email{s39ocard@uni-bonn.de}

\author*[3]{\fnm{Christine} \sur{Heitsch}}\email{heitsch@math.gatech.edu}

\affil[1]{\orgdiv{School of Mathematical and Statistical Sciences}, \orgname{Clemson University}, \orgaddress{\city{Clemson}, \postcode{610101}, \state{SC}, \country{USA}}}

\affil[2]{\orgdiv{Department of Mathematics}, \orgname{University of Bonn}, \orgaddress{\city{Bonn}, \postcode{53115}, \country{Germany}}}

\affil*[3]{\orgdiv{School of Mathematics}, \orgname{Georgia Institute of Technology}, \orgaddress{\city{Atlanta}, \postcode{30332}, \state{GA}, \country{USA}}}


\abstract{\textbf{Background:} 
In the Nearest-Neighbor Thermodynamic Model, a standard approach for RNA secondary structure prediction, the energy of the multiloops is modeled using a linear entropic penalty governed by three branching parameters. Although these parameters are typically fixed, recent work has shown that reparametrizing the multiloop score and considering alternative branching conformations can lead to significantly better structure predictions. However, prior approaches for exploring the alternative branching structures were computationally inefficient for long sequences.
\textbf{Results:} We present a novel algorithm that partitions the parameter space, identifying all distinct branching structures (optimal under different branching parameters) for a given RNA sequence using the fewest possible minimum free energy computations. Our method efficiently computes the full parameter-space partition and the associated optimal structures, enabling a comprehensive evaluation of the structural landscape across parameter choices.
We apply this algorithm to the Archive II benchmarking dataset, assessing the maximum attainable prediction accuracy for each sequence under the reparameterized multiloop model. We find that the potential for improvement over default predictions is substantial in many cases, and that the optimal prediction accuracy is highly sensitive to auxiliary modeling decisions, such as the treatment of lonely base pairs and dangling ends.
\textbf{Conclusion:} 
Our results support the hypothesis that the conventional choice of multiloop parameters may limit prediction accuracy and that exploring alternative parameterizations is both tractable and worthwhile. The efficient partitioning algorithm we introduce makes this exploration feasible for longer sequences and larger datasets. Furthermore, we identify several open challenges in identifying the optimal structure.}




\maketitle

\section{Background}\label{intro}

RNA secondary structure prediction plays a central role in understanding RNA function and regulation~\cite{tinoco-bustamante-99, doudna2000structural}, and despite the emergence of new methods~\cite{schuster1997rna,major2001computational,gardner2004comprehensive,ding2006statistical, mathews-06, leontis2006building, shapiro2007bridging, flamm2008beyond, eddy2014computational, zhao2021review, saman2022rna, Szikszai2022}, experimentalists continue to rely on classical minimum free energy (MFE) predictions~\cite{mathews2006prediction} to extract valuable functional insights. The nearest-neighbor thermodynamic model (NNTM)~\cite{turner2010nndb}  provides a widely used approximation for calculating the folding free energy ($\Delta G$) of an RNA structure by decomposing it into smaller substructures -- such as base pair stacks and loops -- and assigning a score based on the model's  parameters. The total $\Delta G$ is computed as the sum of the contributions from all such components.

An important yet challenging aspect of the NNTM is the treatment of multiloops, which are substructures where three or more helices branch from a central loop. The score for multiloops consists of two parts: an initiation term, interpreted as an entropic penalty, and a stacking term that accounts for favorable intra-loop interactions. The NNTM models the initiation term/entropic penalty using a simple linear function parameterized by three \emph{branching parameters} -- $a$, $b$, $c$ -- representing an initiation penalty, a per-unpaired-nucleotide penalty, and a per-branching-helix penalty, respectively:

\[ \Delta G_{\mathrm{init}} = a + b \cdot [\text{number of unpaired nucleotides}] + c \cdot [\text{number of branching helices}].\]
The linear form of $\Delta G_{\mathrm{init}}$ was initially introduced for computational efficiency, but later comparison~\cite{ward2017advanced} with logarithmic penalty~\cite{mathews1999expanded} based on Jacobson-Stockmayer theory and one based on polymer theory~\cite{aalberts-nandagopal-10} concluded that ``the simplest model is the best''. 

The values of the branching parameters have changed considerably over time and have been (in kcal/mol):
\[T89 = (4.6, 0.4, 0.1) \hspace{1cm} T99 = (3.4, 0, 0.4) \hspace{1cm} T04 = (9.3, 0, -0.6).\]
Unlike many of the NNTM parameters which have been derived experimentally, the values of the T89~\cite{jaeger-turner-zuker-89} and T99~\cite{mathews1999expanded} parameters were learned, i.e. derived by maximizing the MFE prediction accuracy over a set of known structures,  and the T04 parameters were based on experimental data used to parametrize a different function which cannot be accommodated in a dynamic programming algorithm~\cite{mathews-etal-04}. Recent work has shown that reparametrizing the multiloop entropic penalty can yield a markedly improved minimum free energy (MFE) structure, suggesting that it is valuable to explore the resulting alternative structures~\cite{poznanovic2025can}.

Recent work has shown that reparametrizing the multiloop entropic penalty can yield a markedly improved minimum free energy (MFE) structure, but trying to find a ``better'' set of branching parameters leads to overfitting~\cite{bnb}. This suggests that instead of trying to change the model parameters, it is valuable to explore a larger set of alternative structures obtained by reparametrization.~\cite{poznanovic2025can}.

The parameter space used in~\cite{poznanovic2025can} to generate alternative structures was derived from a training set of tRNA and 5S rRNA sequences but was shown to include improved structures for the remaining eight RNA families in the benchmarking dataset. This indicated that generating multiple structures based on reparametrizing the branching initiation penalty can be a valuable source of alternative structures for improving secondary structure prediction.  However, several directions for future research in this direction remained unresolved: determining what constitutes an optimal search space, effects of making other modeling decisions which are standard options in current MFE computation software, and finding the most accurate structure in the structure ensemble.  Moreover, the algorithmic approach used to explore the parameter space was not computationally efficient, particularly for longer RNA sequences.

In this work, we address the computational limitations by introducing a new algorithm that is optimal in the number of MFE evaluations (oracle calls) required to compute the complete parameter-space partition for any given RNA sequence. This partition represents regions of the parameter space that yield distinct optimal secondary structures under different branching parameters. We describe two versions of the algorithm: one that operates with rational precision and recovers all non-zero area regions, and another that uses integer precision—suitable for interfacing with \verb+ViennaRNA+~\cite{lorenz-etal-11}. \verb+ViennaRNA+ allows for parameter adjustments, but scales them by a factor of 100 to integers, and performs MFE calculations in dekacal/mol.

We apply our method to the Archive II dataset and assess the achievable prediction accuracy for each sequence under any linear multiloop model. This was previously achieved for the tRNA and 5S RNA families, but was computationally prohibitive for longer sequences. Our results show that statistically significant improvements are possible for all RNA families with the standard options in \verb+ViennaRNA+, when compaired with the standard MFE and Boltzmann centroid structures. Further analysis reveal that the amount of improvement depends in non-obvious ways on auxiliary modeling choices, such as how lonely base pairs and dangling ends are treated.

These results underscore the value of continued development of this framework,  but the problem of how to identify the most accurate structure remains. In working towards solving the problem of how to find the 'needle in the haystack', we identified some problems to be addressed, including (a) density of target structure pairings, (b) treatment of pseudoknots, (c) relevance of coaxial stacking and other more sophisticated branching energy models, (d) pairings which are \textit{close} to the target helix/stem, but disjoint from it (even under +/- 1 nt skew), etc.

\section{Preliminaries} \label{prelims}

\subsection{Branching Signatures} 

An RNA secondary structure is a set of complementary base pairings (hydrogen bonds) formed within a single-stranded RNA molecule, which determine its shape in two dimensions. The main structure elements are helices (regions of consecutive base pairs), hairpin loops (single-stranded loop at the end of a helix), internal loops (unpaired bases between two helices),
bulges (unpaired bases on one side of a helix), and multiloops (unpaired regions connecting three or more helices).

When the non-branching parameters  in the model are fixed, the total folding free energy change of a structure $S$ is given by
\begin{equation}\Delta G(S) = ax + by + cz + w, \end{equation}
where  $(a,b,c)$ are the branching parameters and $x$, $y$, and $z$ are the total number, respectively, of multiloops as well as of single-stranded bases and of helices in those loops, and $w$ is the residual free energy from all the other structural components.
We refer to the quadruple $(x,y,z,w)$ as the \emph{branching signature} or sometimes just the \emph{signature} of the structure. Given a sequence \seq, we use the term \emph{signatures of} $s$ to refer to the set of the signatures of all the secondary structures for \seq. Note that by definition,  $x$, $y$, and $z$ are nonnegative integers, and $w$ is a rational number as it depends on rational parameters, usually specified with precision of two decimal places~\cite{turner2010nndb}.

\subsection{Parameter Space Partition}

The number of possible RNA secondary structures for a given sequence grows exponentially with its length, but is finite. Consequently, the three-dimensional branching parameter space can be partitioned into a finite number of regions, each corresponding to a distinct optimal branching signature. In~\cite{regions} it was shown that these regions are convex sets such that any parameter combination from a given region yields the same optimal branching signature. Points located on the boundaries between regions correspond to parameter values where multiple optimal signatures coexist, each associated with a neighboring region.

This partitioning of the parameter space arises from the geometric properties of the so-called RNA branching polytope, a construct that encapsulates all possible branching signatures for a given RNA sequence~\cite{pmfe_chapt}. The geometric framework for finding the parameter space partition for a given sequence allows for a comprehensive analysis of how variations in branching parameters affect RNA secondary structure predictions. 

Analysis of the parameter space partition~\cite{polystats} showed that the regions are much thinner in the $b$-direction than in the other two. In other words, the MFE prediction is most sensitive under perturbation of the $b$ parameters. Subsequent work in~\cite{bnb} focused on restrictions of the aforementioned regions to a fixed $b=b_0$ plane, with the goal of understanding how the geometry of the target branching signature regions differ for the tRNA and 5S RNA families. 

Results showed that each of these two families has a characteristic target region geometry, which is distinct from the other and significantly different from their own dinucleotide shuffles. This explained prior results which showed that modifying the branching parameters can significantly enhance secondary structure prediction accuracy for the two individual RNA families. However, when attempting to apply a single set of optimized parameters across multiple families, the overall improvement diminishes compared to family-specific optimizations. In other words, each RNA family benefits from distinct branching parameters, as was previously observed in~\cite{polystats}.

\subsection{Parameter Plane Partition}

Since we are primarily interested in analyzing perturbations of the $a$ and $c$ parameters, here we introduce notation for the regions in the parameter space decomposition as restricted to a fixed $b=b_0$ plane.

\begin{definition}
    For a signature $\sigma$ of a sequence \seq, the region of $\sigma$ in the $b = b_0$ plane is 
    \[\reg(\sigma) = \{(a,c) : (a, b_0, c, 1) \cdot \sigma \leq (a,b_0,c,1) \cdot \sigma', \text{for every signature $\sigma'$ of $s$} \}.
    \] 
\end{definition}

The region of $\sigma$ depends on the sequence $s$ and the parameter $b_0$, but since in the discussion that follows, we assume that both are fixed, for simplicity,  we suppress this dependency in the notation of $\reg(\sigma)$ as well as other terms we define later. 

\begin{definition}
    Let $\Sigma$ be the set of signatures for a sequence \seq. The set \[\Pi = \{\reg(\sigma) : \reg(\sigma) \neq \emptyset, \sigma \in  \Sigma\}\] is the partition of the $(a,c)$-parameter plane $b=b_0$ for the sequence \seq. 
\end{definition}

If the set $\Sigma$ of all signatures is known, one can compute $\reg(\sigma)$ as an intersection of half-planes. Namely, a signature $\sigma =(x,y,z,w)$ is more optimal than another signature $\sigma'=(x',y',z',w')$ on the half-plane $H(\sigma \leq \sigma')$ given by:
\begin{equation} \label{eq:halfplane_eq} H(\sigma \leq \sigma'): \; \; \; a(x-x') + c(z-z') \leq b_0(y'-y) + (w'-w),\end{equation}
and therefore
\begin{equation} \label{eq:reg_formula} \reg(\sigma) = \bigcap_{\sigma' \in \Sigma} H(\sigma \leq \sigma').
\end{equation}

Note that~\eqref{eq:reg_formula} still applies if $\Sigma$ is taken to be the (much smaller) set of signatures that are optimal for at least one combination of parameters in the $b=b_0$ plane, and this significantly reduces the complexity of computing $\reg(\sigma)$. In Section~\ref{sec:discussion}, we present a figure of the partitions for a couple of example sequences we discuss.

Clearly, identifying a signature requires at least one MFE calculation. Since the computational time for this step increases cubically with sequence length~\cite{zukermfealg, lyngso-zuker-pedersen-99a}, we propose algorithms for constructing the partition $\Pi$ that minimize the number of required MFE computations. We note, that for analyses focused on highly localized regions, one could compute all optimal structures for a given parameter combination $p$ by employing the algorithm that enumerates all suboptimal structures within a specified energy interval around the MFE structure~\cite{wuchty1999complete}. In such cases, the regions surrounding the point $p$ exhibit the following organization. Let $\mathrm{conv}(A)$ denote the convex hull of the set $A$.

\begin{theorem} \label{thm: convhull} Let $S = \{(x_i, y_i, z_i, w_i) \; : \; i =1, \dots, n \}$ be the set of all signatures of the MFE structures for the parameters $p =(a_0, b_0, c_0, 1)$. Then $p \in \reg(x_i, y_i, z_i, w_i), i =1, \dots, n$ and these regions are geometrically arranged around $p$ in the same order the vertices of the polygon $P = \mathrm{conv}(\{(x_i, z_i) \; : \; i =1, \dots, n \})$ are arranged in the plane.
\end{theorem}

\begin{proof}
    Let $\mathrm{NF}(P)$ be the normal fan of $P$. We will show the stronger statement that the rays in $\mathrm{NF}(P) + (a_0,c_0)$ (formed by the outer normal vectors for the edges of $P$) contain the boundary segments of regions in $b=b_0$ incident with $p$.

Let $(a,c)$ be a point such that $p+(a, 0, c, 0)$ is in a region incident with $p$ and let $(x,y,z,w) \in S$ be the signature of the MFE structures for this region. Then 
\begin{equation} \label{dot} (p+(a, 0, c, 0))\cdot (x,y,z,w) \leq (p+(a, 0, c, 0))\cdot (x',y',z',w') \end{equation}
for every $(x',y',z',w') \in S$. By the definition of $S$, we have \[p \cdot (x,y,z,w) = p \cdot (x',y',z',w')\] for every $(x',y',z',w') \in S$. Therefore,~\eqref{dot} implies 
\begin{equation} \label{dot2} (a, 0, c, 0)\cdot (x,y,z,w) \leq (a, 0, c, 0)\cdot (x',y',z',w')\end{equation} for every $(x',y',z',w') \in S$, which is equivalent to 
\begin{equation} \label{dot3} (a, c)\cdot (x,z) \leq (a, c)\cdot (x',z').\end{equation} This means that $(a, c)$ is in the cone  of $\mathrm{NF}(P)$ corresponding to $(x,z)$. Therefore, if $p+(a, 0, c, 0)$ belongs to two regions incident with $p$, $(a, c)$ is on the ray which represents the boundary of two 2d cones in $\mathrm{NF}(P)$.
\end{proof}

As a direct corollary of the proof of Theorem~\ref{thm: convhull}, we have that, unlike in an arbitrary partition of the plane into convex polygonal regions, a point shared by more than two regions must be a vertex for all of them (Figure~\ref{lemmaimage}). Otherwise, we get a contradiction to the fact that the boundaries between neighboring regions which share an edge give the directions of the outer normal vectors for the polygon $P$ defined by the associated signatures.

\begin{corollary} \label{cor:vertex} If a point is shared by more than two regions, it must be a vertex for all of them.
\end{corollary}

\begin{figure} 
\centering
\includegraphics[width=.7\linewidth]{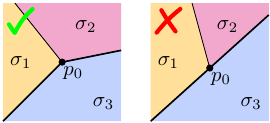}
\caption{A point shared by more than two regions is necessarily a vertex for all of them}
\label{lemmaimage}
\end{figure}

\section{Parameter Plane Partition Algorithms}\label{sec:algs}

In this section, we present two algorithms for computing the partition of the parameter plane $\Pi$, up to zero-area regions, under the assumption that the MFE optimizer accepts rational or integer inputs.

Typically, we are interested in analyzing structures that arise within a specific subset of parameters values, so we present  algorithms for computing the partition restricted to a convex polygonal frame $F$. By setting $F$ sufficiently large, one can recover the partition of the whole parameter plane, if desired.

The algorithms repeatedly perform MFE calculations, followed by the computation of the corresponding branching signature of the MFE structure. We refer to this combined process as querying an oracle.

\begin{definition}
Let $D$ be a subset of the $(a,c)$-parameter plane. An $\oracle$ is a function $D \rightarrow \Sigma$, whose value at $(a,c) \in D$ is a signature of a structure which is optimal for the parameters $(a,b_0,c)$, i.e.
\[ (a,c) \in D, \oracle(a,c) = \sigma \implies (a,c) \in \reg(\sigma). \]
\end{definition}
Note that while multiple optimal signatures may exist for a given set of parameters, we assume that the oracle is consistent -- i.e., it always returns the same signature. In practice, this consistency arises because implementations of the MFE algorithm employ deterministic procedures during the traceback phase used to recover an MFE structure.

\subsection{Decimal Precision Parameter Partition Algorithm}\label{fixed}

Let $F$ be a convex polygonal frame in the $b=b_0$ plane. In this subsection we describe an algorithm for finding the partition of the parameter space inside $F$ for a given sequence when we have an $\oracle$ which takes inputs from $D = \mathbb{Q} \times \mathbb{Q}$. 

\begin{definition}
 Let $\sigma$ and $\sigma'$ be two signatures and let $p$ be a point in the parameter space. We say that $\sigma$ dominates $\sigma'$ at $p$ if $\sigma \cdot p \leq  \sigma' \cdot p$. 
\end{definition}

The set of all points $(a, b_0, c)$ where $\sigma = (x,y,z,w)$ dominates another signature $\sigma'=(x',y',z',w')$ is the half-plane $H(\sigma \leq \sigma')$ defined by~\eqref{eq:halfplane_eq}. Therefore,  if several signatures are discovered by making oracle calls, one can find the partition of the plane or a given frame into sets on which one signature is better than the rest, using intersections of half-planes in the following way.
\begin{definition} Let $\Sigma$ be a set of signatures. For a signature $\sigma$, the dominating set of $\sigma$ relative to $\Sigma$ is 
    \[\dom_{\Sigma}(\sigma) = \bigcap_{\sigma' \in \Sigma} H(\sigma \leq \sigma').\] If $F$ is a convex polygonal frame, the dominating set of $\sigma$ in $F$ relative to $\Sigma$ is $ \dom_{\Sigma}(\sigma) \cap F.$
    \end{definition}

It follows from the definitions that discovering new signatures leads to a better approximation of the region of $\sigma$, i.e., if $\Sigma \subseteq \Sigma'$, then $\dom_{\Sigma}(\sigma) \supseteq \dom_{\Sigma'}(\sigma) \supseteq \reg(\sigma)$. This relationship forms the basis of our algorithm (Algorithm~\ref{alg1}). To identify all signatures and their corresponding regions within the frame $F$, we proceed iteratively. Each time a new signature is discovered, the frame is repartitioned into dominating sets based on the expanded set $\Sigma$. We then query the oracle at the vertices of these dominating sets, one at a time. If the oracle returns a previously discovered signature, we confirm that vertex as part of the final partition. Otherwise, the newly discovered signature is added to $\Sigma$, and the partition is updated accordingly (Figure~\ref{fig:algo_fig}). 
\begin{figure}
    \centering
    \includegraphics[width=0.7\linewidth]{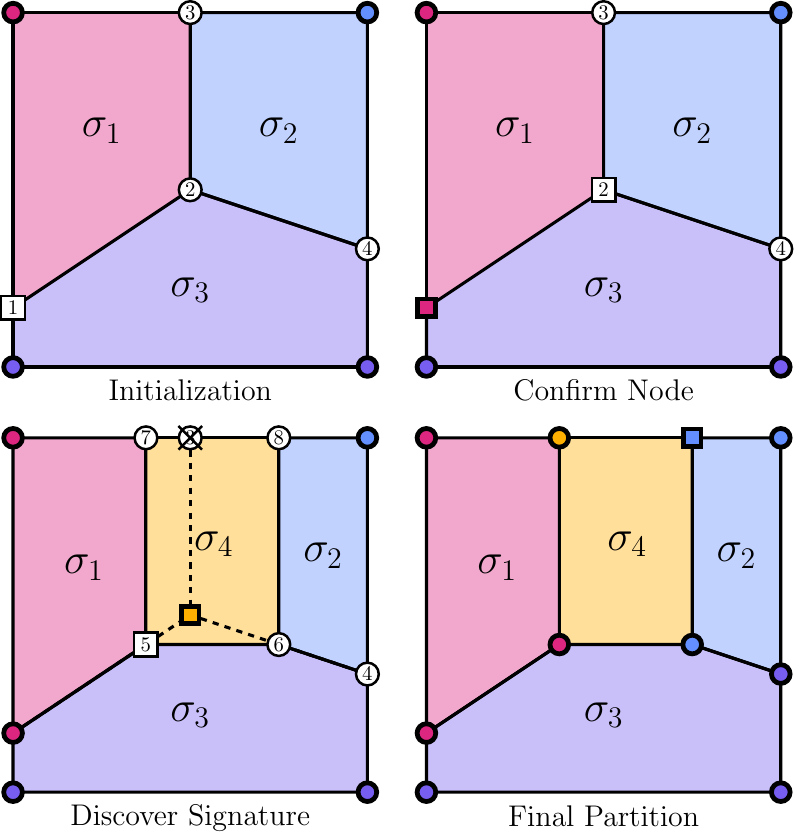}
    \caption{An artificial run of Algorithm~\ref{alg1}. Filled circles are confirmed nodes with color matching the corresponding signature obtained from the oracle. Empty circles are the points in $N$ numbered by the order to be checked. A numbered square represents the next point to be checked and a filled square is the most recently checked point. During the initialization step (top left), the signatures and nodes on the corner of the frame are initialized. When the oracle is queried at $1$, it returns one of the known signatures ($\sigma_1$), so $1$ is confirmed as a node in the final partition (top right). When the oracle is queried at $2$, it returns a new signature ($\sigma_4$), and the partition is updated, which results in 4 new vertices being added and one vertex being deleted from the list $N$ for processing (bottom left). The algorithm ends when all vertices of the current partition are confirmed (bottom right)}
    \label{fig:algo_fig}
\end{figure}

\begin{algorithm}
    \caption{Decimal Precision Parameter Partition Algorithm}\label{alg1}
    \hspace*{\algorithmicindent} \textbf{Input:} An RNA sequence $s$, value $b_0$ for scoring single-stranded nucleotides in the branching loops, a (convex polygonal) frame $F$ in the $(a,c)$-plane\\
    \hspace*{\algorithmicindent} \textbf{Output:} Partition $\Pi$ of the interior of the frame $F$ into regions with different optimal signatures; list $\Sigma$ of signatures which correspond to those regions\\
    \begin{algorithmic}
    \Function{partition1}{$F$}
    \Statex \Comment{Initialize signatures}
    \State $\Sigma \gets \{ \mathrm{oracle}(v) \mid v \text{ is  a vertex of } F\}$
    \Statex \Comment{Initialize nodes}
    \State $N \gets \{(v, \text{`confirmed'}) \mid v \text{ is  a vertex of } $F$\}$ 
    \Statex \Comment{Initialize partition}
    \For{$\sigma$ in $\Sigma$}
        \State add all other vertices of $\dom_{\Sigma}(\sigma) \cap F$ to $N$ with `unconfirmed' flag
    \EndFor
    \Statex \Comment{Search for new signatures and update partition while possible}
    \For{$v$ in $N$ with `unconfirmed' flag} 
       \State $\sigma \gets \oracle(v)$
        \If{$\sigma \in \Sigma$}
            \State change flag to `confirmed'
        \Else
            \State $\Sigma \gets \Sigma \cup \{\sigma\}$
\Statex \Comment{Include the vertices of the new dom set}
            \If{a vertex of $\dom_{\Sigma}(\sigma) \cap F$ is not in $N$}
                \State add it to $N$ with an unconfirmed flag
            \EndIf
\Statex \Comment{Exclude points which are no longer vertices of dom sets}
            \If{a point in $N$ is in the interior of $\dom_{\Sigma}(\sigma)$}
                \State remove it from $N$ 
            \EndIf
        \EndIf
    \EndFor
    \State $\Pi \gets \{\dom_{\Sigma}(\sigma) \cap F: \sigma \in \Sigma \}$ 
    \State \Return ($\Pi$, $\Sigma$)
    \EndFunction
    \end{algorithmic}
\end{algorithm}

The following lemma shows that the set $N$ in Algorithm~\ref{alg1} indeed contains the set of vertices of the dominating sets at that stage, i.e., the way we are updating the partition is valid. Namely, we show that when a new signature $\sigma$ is discovered and the dominating sets are updated, to keep  an accurate information about their vertices, in addition to taking into account the vertices of $\dom_{\Sigma}(\sigma) \cap F$ it suffices to only exclude the vertices which are in the interior of dominating set for the new signatures.

\begin{lemma} \label{lemma1} The set $N$ in Algorithm~\ref{alg1} is the union of the vertices of $\dom_{\Sigma}(\sigma) \cap F$, for all $\sigma \in \Sigma$. 
\end{lemma}

\begin{proof} The claim is clearly true at the initialization step because $N$ is formed from precisely the vertices of the dominating sets for the signatures discovered at the corners of the frame $F$. Let $\sigma'$ be a new signature discovered, after which the set $\Sigma$ is updated to $\Sigma' = \Sigma \cup \{\sigma'\}$, and $N$ is updated to $N'$. 

The dominating sets have disjoint interiors, so clearly a vertex of $\dom_{\Sigma}(\sigma) \cap F$ which is in the interior of $\dom_{\Sigma'}(\sigma')$ is no longer a vertex of an updated  dominant set. A point in $N$ which is not in the interior of $\dom_{\Sigma'}(\sigma')$ is a vertex of a dominant set in the updated partition because \begin{equation}\label{eq:newdom} \dom_{\Sigma'}(\sigma) =  \dom_{\Sigma}(\sigma) \setminus \mathrm{interior}(\dom_{\Sigma'}(\sigma')).\end{equation} This shows that the  elements of $N \cap N'$ are indeed vertices of the updated partition. 

By definition $N' \setminus N$ is a subset of the vertices of $ \dom_{\Sigma'}(\sigma') \cap F$, so $N'$ is a subset of the vertices of the updated partition. To show equality, we argue that for $\sigma \in \Sigma$, each vertex of $\dom_{\Sigma'}(\sigma) \cap F$ is either a vertex of $\dom_{\Sigma}(\sigma) \cap F$ or is a vertex of $ \dom_{\Sigma'}(\sigma') \cap F$, and is therefore in $N'$. Namely, suppose that $\dom_{\Sigma'}(\sigma) \cap F$ has a vertex $v$ which is not a vertex of $\dom_{\Sigma}(\sigma) \cap F$. Then because of~\eqref{eq:newdom}, $v$ is on the boundary of $\dom_{\Sigma'}(\sigma') \cap F$. If $v$ is on the boundary of $F$, then it is necessarily a vertex of $\dom_{\Sigma'}(\sigma') \cap F$, by Corollary~\ref{cor:vertex}. 
\end{proof}

We show that when a vertex is confirmed, it is indeed a vertex in the target partition. This explains why we can stop searching for new regions around that vertex.

\begin{lemma}\label{lem:nodes}
 When a vertex is confirmed in Algorithm~\ref{alg1}, it is a node in the target partition.  
\end{lemma}
\begin{proof} Let $v$ be a point which is in $N$ at some stage of the algorithm and gets marked `confirmed'. If $v$ is a vertex of $F$, it is clearly also a vertex of at least one region in the target partition. Otherwise, $v$ is added to $N$ with an `unconfirmed' flag as a vertex of a dominant set for a new discovered signature.  As new signatures are discovered, the set $N$ gets updated so it always contains only the vertices of the dominant sets at that stage (Lemma~\ref{lemma1}). 

Let $\Sigma$ be the set of discovered signatures when the oracle is called with input $v$. At that point the value $\sigma \cdot v$ is minimized over $\Sigma$ by the signatures $\sigma$ whose dominant sets contain $v$. If no new signature is discovered by calling the oracle, then  that is the absolute minimum over all possible signatures. Therefore, if $v$ is a vertex of $F \cap \dom_{\Sigma}(\sigma)$ for $\sigma \in \Sigma$ then $v \in \reg(\sigma)$ and is a node, because $\reg(\sigma) \subseteq \dom_{\Sigma}(\sigma)$.
\end{proof}

\begin{theorem}[Correctness] Algorithm~\ref{alg1} finds all signatures $\sigma$ for which $\reg(\sigma) \cap F$ has non-zero area.
\end{theorem}
\begin{proof}
Let $\sigma$ be a signature found by Algorithm~\ref{alg1}. All the vertices of  $\dom_{\Sigma}(\sigma) \cap F$ are sent to the oracle at some point in the algorithm (either because the vertex was confirmed previously or was added to $N$ when $\sigma$ was discovered and confirmed later). This means that the minimum value of the energy function for those parameters has been confirmed to be achieved by the signatures of the dominating sets the vertex belongs to. Since $\sigma$ is optimal for all the vertices of $\dom_{\Sigma}(\sigma)$, by convexity, it is also optimal for all interior vertices of $\dom_{\Sigma}(\sigma)$. Therefore, $\dom_{\Sigma}(\sigma) \subseteq \reg(\sigma)$. Since the reverse inclusion is always true, we have $\dom_{\Sigma}(\sigma) = \reg(\sigma)$ for every signature $\sigma$ found by Algorithm~\ref{alg1}.

Let $\sigma_0$ be a signature for which $\reg(\sigma_0) \cap F$ has non-zero area. Then there is a discovered signature $\sigma \neq \sigma_0$ such that $\reg(\sigma_0) \cap \reg(\sigma_0)$ also has non-zero area, which is a contradiction with the fact that distinct regions can only intersect at a vertex or along a segment.
\end{proof}

\begin{corollary}[Complexity]\label{cor:complexity}
 The number of oracle calls in Algorithm~\ref{alg1} is $n + m$, where $n$ is the number of regions discovered and $m$ in the number of vertices in the partition.
\end{corollary}
\begin{proof}
This follows from the fact that when an oracle is called, either a new signature is discovered, or by Lemma~\ref{lem:nodes}, a node in the partition is confirmed.   
\end{proof}

\subsection{Integer Precision Parameter Partition Algorithm}\label{sec:integer}

Let $F$ be a convex polygonal frame in the $b=b_0$ plane. In this section we describe an algorithm for finding the partition of the parameter space inside $F$ for a given sequence when we have an $\oracle$ which takes inputs from the integer lattice $\Lambda = \mathbb{Z} \times \mathbb{Z}$. 

 In this case one can find dominating sets as before, but querying the oracle at their vertices is usually impossible because they are rarely lattice points. To circumvent this problem, when we need to confirm a vertex $(a,c)$ in $N$, which is not in $\Lambda$, we call the oracle for the set of lattice points around it:
$\{(\lfloor a \rfloor, \lfloor b \rfloor), (\lfloor a \rfloor, \lceil b \rceil), (\lceil a \rceil, \lfloor b \rfloor), (\lceil a \rceil, \lceil b \rceil)\}$ (Algorithm~\ref{alg2}). If a new signature is found for any of these lattice points, the set $\Sigma$ is updated and the partition is recomputed. If no new signature is found around the vertex $(a,c)$, we call that vertex confirmed. Unlike in Algorithm~\ref{alg1}, in this case there is a possibility that a signature with a small non-zero region is not discovered. However, spot-checking the output against the results obtained by computing the signature of every lattice point in the rectangle Grec for several RNA sequences from the Archive II dataset, revealed that no non-zero regions were missed (we give a description of the Archive II dataset and Grec in Section~\ref{sec:results}.)

\begin{algorithm}[htb]
    \caption{Integer Precision Parameter Partition Algorithm}\label{alg2}
    \hspace*{\algorithmicindent} \textbf{Input:} An RNA sequence $s$, value $b_0$ for scoring single-stranded nucleotides in the branching loops, a (convex polygonal) frame $F$ in the $(a,c)$-plane\\
    \hspace*{\algorithmicindent} \textbf{Output:} Partition $\Pi$ of the interior of the frame $F$ into regions with different optimal signatures; list $\Sigma$ of signatures which correspond to those regions\\
    \begin{algorithmic}
    \Function{partition2}{$F$}
    \Statex \Comment{Initialize signatures}
    \State $\Sigma \gets \{ \mathrm{oracle}(v) \mid v \text{ is  a lattice point near a vertex of } F\}$
    \Statex \Comment{Initialize nodes}
    \State $N \gets \{(v, \text{`confirmed'}) \mid v \text{ is  a vertex of } $F$\}$ 
    \Statex \Comment{Initialize partition}
    \For{$\sigma$ in $\Sigma$}
        \State add all vertices of $\dom_{\Sigma}(\sigma) \cap F$ to $N$ with `unconfirmed' flag
    \EndFor
    \Statex \Comment{Search for new signatures and update partition while possible}
    \For{$v$ in $N$ with `unconfirmed' flag} 
        \State $\mathrm{Sigs}_v \gets \{\oracle(p) \mid p \text{ lattice point near } v \}$
        \If{$\mathrm{Sigs}_v \subseteq \Sigma$}
            \State change flag of $v$ to `confirmed'
        \Else
            \State $\Sigma \gets \Sigma \cup \mathrm{Sigs}_v$
\Statex \Comment{Include the vertices of the new dom set}
            \If{a vertex of $\dom_{\Sigma}(\sigma) \cap F$ is not in $N$}
                \State add it to $N$ with an unconfirmed flag
            \EndIf
\Statex \Comment{Exclude points which are no longer vertices of dom sets}
            \If{a point in $N$ is in the interior of $\dom_{\Sigma}(\sigma)$}
                \State remove it from $N$ 
            \EndIf
        \EndIf
    \EndFor
    \State $\Pi \gets \{\dom_{\Sigma}(\sigma) \cap F: \sigma \in \Sigma \}$ 
    \State \Return ($\Pi$, $\Sigma$)
    \EndFunction
    \end{algorithmic}
\end{algorithm}

\newpage

\section{Results}\label{sec:results}

In this section, we present an analysis of the computational complexity of our method using a benchmarking dataset. Previous work on this dataset~\cite{poznanovic2025can} demonstrated that statistically significant improvements in prediction accuracy can be achieved by considering an ensemble of branching structures as the basis for prediction. We replicate those findings using a different implementation of the MFE optimization and further investigate how variations in the NNTM  model influence the set of structures generated through branching reparameterization.

\subsection{Complexity} 
We implemented Algorithm~\ref{alg2} in which one oracle call consists of two steps: running \verb+ViennaRNA+~\cite{lorenz-etal-11} to compute an MFE structure for modified branching parameters $(a,b,c)$ followed by using a structure parser to find the corresponding signature. The code is freely available at: \url{https://github.com/gtDMMB/ParamPartAlg}. Our computations were done on a RHEL7.9 system with an Intel(R) Xenon(R) Gold 6240 CPU @ 2.60GHz.

We computed the partition inside the  Grec frame identified in~\cite{poznanovic2025can} for each sequence in the Archive II benchmarking dataset~\cite{sloma2016exact, mathews-19} from the Mathews Lab (U Rochester). For analyzing the parameter space in~\cite{poznanovic2025can}, it was convenient to make a change of variables: $\bar{a} = a + 3c$ (as each multiloop has at least 3 helices, $\bar{a}$ represents the minimum penalty per multiloop). In these coordinates, the frame Grec is the rectangle $300 \leq \bar{a} \leq 800$, $-350 \leq c \leq 100$ in the $b=0$ plane (in dekacal/mol). The Archive II dataset contains 3948 sequences from 10 families, with summary statistics given in Table~\ref{tab:dataset}. Since the set sizes vary significantly, for visual representations of data that follow, we group the families with fewer than 100 sequences together.

\begin{table}[h]
\caption{Archive II dataset}\label{tab:dataset}
\begin{tabular}{@{\extracolsep\fill}lcccc}
\toprule%
& & \multicolumn{3}{@{}c@{}}{sequence length} \\\cmidrule{3-5}%
family & {set size} & min& median& max\\
\midrule
tRNA  &  557 & 54 & 76.0 & 93\\
5S rRNA &  1,283 & 102 & 119.0 & 135\\
SRP RNA &  928 & 28 & 117.5 & 533\\
RNase P &  454 & 120 & 330.0 & 486\\
tmRNA &  462 & 102 & 363.0 & 437\\
16S rRNA domains &  88 & 73 & 399.0 & 655\\
group I intron &  98 & 210 & 431.0 & 736\\
telomerase &  37 & 382 & 445.0 & 559\\
23S rRNA domains  &  30 & 242 & 446.5 & 708\\
group II intron &  11 & 619 & 739.0 & 780\\
\botrule
\end{tabular}
\end{table}

\begin{figure}
\centering
\includegraphics[width=\linewidth]{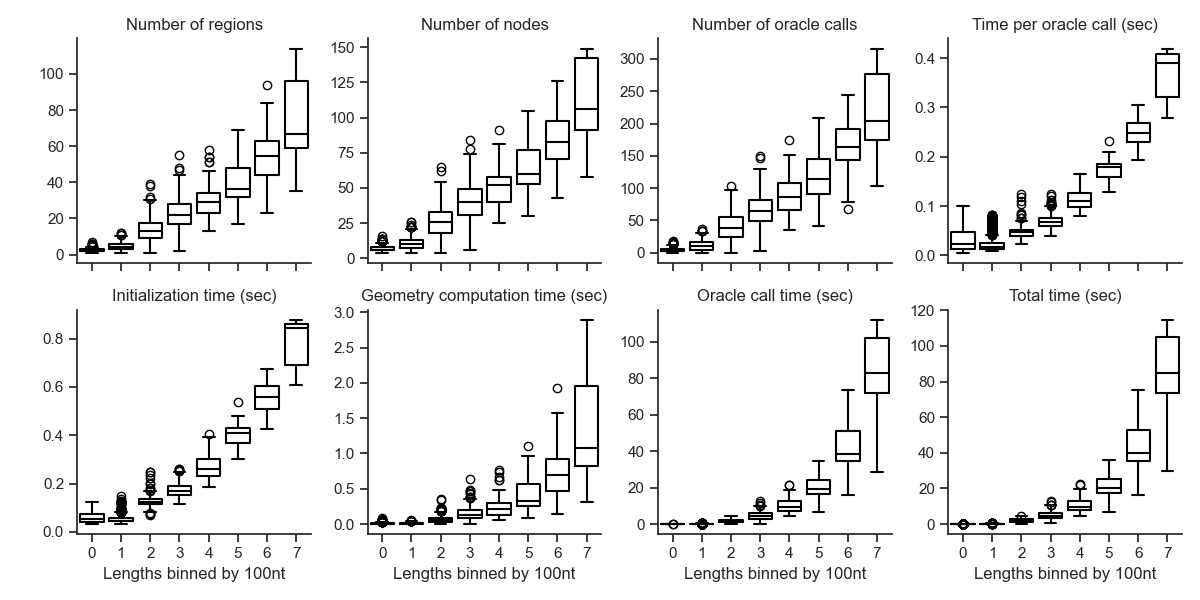}
\caption{Computational complexity of Algorithm 2 on the Archive II dataset, with sequences binned by length in 100nt intervals. The number of sequences in each bin is: 0–99nt (814), 100–199nt (1,579), 200–299nt (379), 300–399nt (911), 400–499nt (170), 500–599nt (58), 600–699nt (28), and $\geq700$nt (9)}
\label{fig:timingboxplots}
\end{figure}

Complexity measurements on this dataset for the Grec frame are summarized  in Figure~\ref{fig:timingboxplots}. As expected, the (avg) time per oracle call increases with sequence length (the MFE algorithm complexity is $O(n^3)$, where $n$ is sequence length~\cite{zukermfealg}). Since both the number of multiloops and the maximum number of branches scale linearly with sequence length, the total number of regions grows as $O(n^2)$. Empirically, each region introduces approximately two new nodes on average, so the node count also follows a quadratic trend.  Consequently, the total number of oracle calls exhibits $O(n^2)$ growth, and the total time grows at a rate $O(n^5)$ (Figure~\ref{fig:timecurves}),  but the coefficients are quite small, so this keeps the computation highly tractable at the lengths we consider.

\begin{figure}[htb]
\centering
\includegraphics[width=\linewidth]{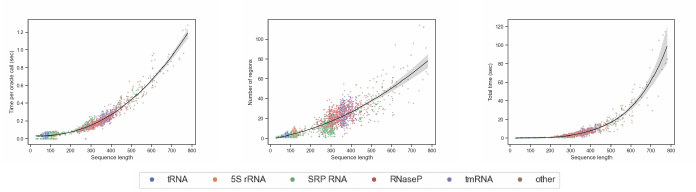}
\caption{Time per oracle call, number of regions, and total compute time for finding the partition of Grec, with the best fitted cubic, quadratic and degree 5 polynomial, respectively}
\label{fig:timecurves}
\end{figure}

The total compute time for the sequences of length $\leq 200$ (60\% of the dataset) is well under one second, with $\max =0.68$ sec. For longer sequences,  the break-up by allocations of the total compute time shows that the time complexity is primarily affected by the number of oracle calls -- 92\% of the total time for the sequences of length $>200$ (40\% of the dataset) is spent on oracle calls after initialization. In other words, the time spent on initialization and geometry (partition updates) is negligible -- even for the 9 sequences of length $\geq 700$, this time is 1.95 sec. This supports the premise that an efficient algorithm for finding the partition should minimize number of MFE predictions.

\subsection{Assessment of Alternative Configurations}

The branching structures from Grec were compared against the target structure for each sequence using the F1 score. For each sequence we found the most accurate branching structure in the sample and we assessed how prediction accuracy would be improved if we were able to identify this structure as an outcome of a prediction method. Figure~\ref{fig:scatetrplots} shows the distribution of the improvements along with other sequence characteristics that are known to affect prediction accuracy: length, pairing density, and GC content.

\begin{figure}
\centering
\includegraphics[width=\linewidth]{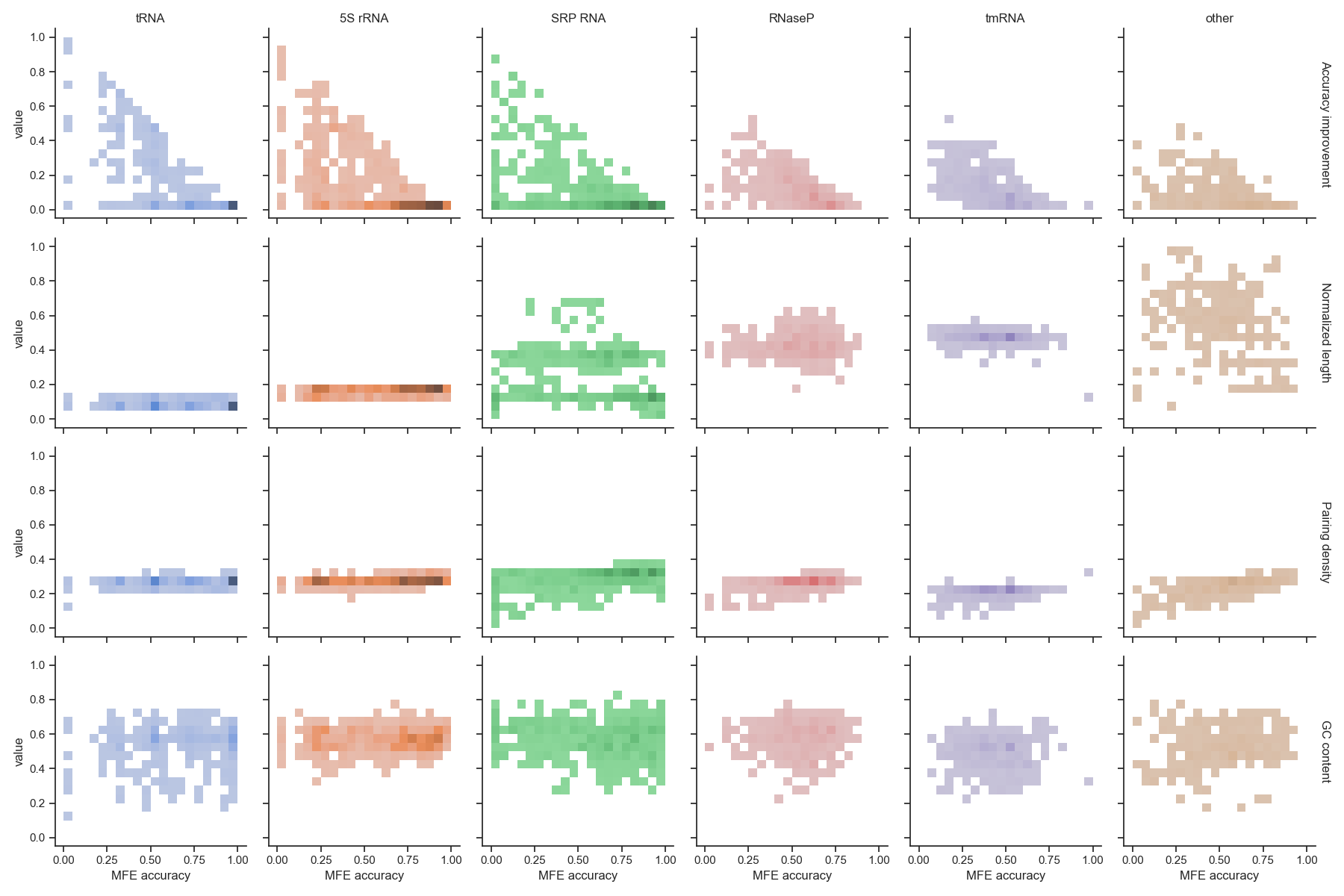}
\caption{Distribution of accuracy improvement when Best Grec structure is considered along with sequence length, pairing density and GC content. Length has been normalized by $\max = 780$}
\label{fig:scatetrplots}
\end{figure}

As can be seen from Table~\ref{tab:mn2vsgn2}, the accuracy of the best Grec structure is statistically significantly higher than the MFE accuracy for all families, under a Bonferroni correction. This confirms the conclusions from~\cite{poznanovic2025can}, where such comparisons were made when slices were generated using different MFE prediction software.

\begin{table}[htb]
\caption{Grec best vs MFE prediction accuracy. The p-value from a paired t-test is given for each family and bolded when significant at a Bonferroni-corrected threshold
$0.005$.}\label{tab:mn2vsgn2}
\begin{tabular}{@{\extracolsep\fill}lccccc}
\toprule%
& \multicolumn{2}{@{}c@{}}{MFE} & \multicolumn{2}{@{}c@{}}{Grec best} & \\\cmidrule{2-3}\cmidrule{4-5}%
family & avg & std & avg & std & p-value\\
\midrule
tRNA  &  0.68	& 0.25	& 0.81	& 0.17 & \textbf{0.0000}\\
5S rRNA &  0.62 & 0.26 & 0.71 & 0.20 & \textbf{0.0000}\\
SRP RNA &  0.59	& 0.30	& 0.67	& 0.25 & \textbf{0.0000}\\
RNase P &  0.54 & 0.16 & 0.66 & 0.13 & \textbf{0.0000}\\
tmRNA &  0.43	& 0.15	& 0.54	& 0.13 & \textbf{0.0000}\\
16S rRNA domains &  0.51 & 0.23 & 0.60 & 0.21 & \textbf{0.0000}\\
group I intron &  0.50 & 0.21 & 0.60 & 0.19 & \textbf{0.0000}\\
telomerase &  0.49	& 0.15	& 0.58	& 0.11 & \textbf{0.0000}\\
23S rRNA domains  &  0.66 & 0.14 &  0.74 & 0.09 & \textbf{0.0000}\\
group II intron &  0.29	& 0.12	& 0.38	& 0.14 & \textbf{0.0002}\\
\botrule
\end{tabular}
\end{table}

It is worth considering how best Grec compares with other prediction methods which generate multiple structures. The current standard for alternative structure generation is stochastic sampling from the Boltzmann distribution~\cite{mccaskill-90, ding-lawrence-03}. \verb+ViennaRNA+  uses its computed partition function to generate a centroid secondary structure,  defined as the secondary structure whose base pairs each have marginal probability $p > 0.5$ in the Boltzmann distribution~\cite{ding-chang-lawrence-05}.  

While for most of the sequences the best accuracy is still the MFE or the centroid structure, we find that Grec best is an improvement over both for 47\% of the sequences. The breakdown per family is given in Figure~\ref{fig:scatter}. Furthermore, this improvement exceeds 0.05 in 34\% of sequences and 0.10 in 24\% of sequences. These results demonstrate that generating alternative structures via branching reparametrization yields a new rich ensemble that can substantially enhance prediction accuracy.

\begin{figure}
\centering
\includegraphics[width=.6\linewidth]{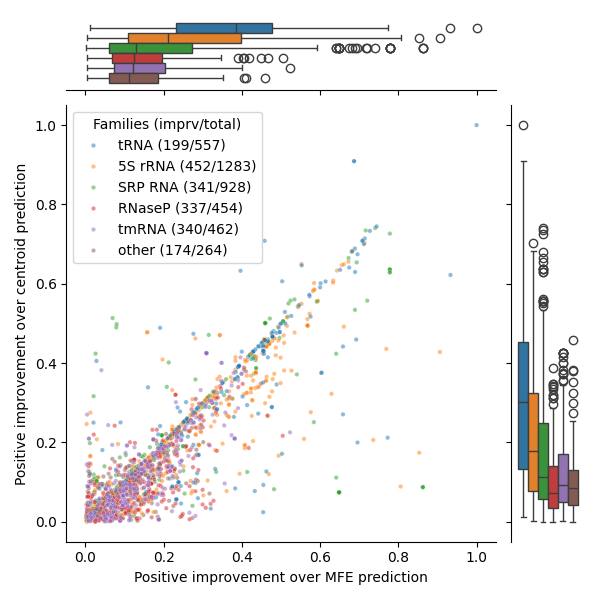}
\caption{Positive improvements of the best Grec structure over the MFE prediction and the Boltzmann centroid, with the marginal distributions. The improvements pictured are statistically significant according to a paired t-test at a Bonferroni-corrected threshold $0.008$.}\label{fig:scatter}
\end{figure}

The predictions discussed thus far were generated using the default options on the 
\verb+ViennaRNA+ web server~\cite{gruber2008vienna}: noLP (no lonely pairs) and the d2 dangling end model (dangling energies are included on both sides of a helix in any case). However, small changes in the model can change the prediction significantly~\cite{gtfold} and we find that toggling the noLP option off and using the d1 dangling end model (unpaired bases can participate in at most one dangling end) affect the best Grec predictions as well. 

\begin{figure}
\centering
\includegraphics[width=.9\linewidth]{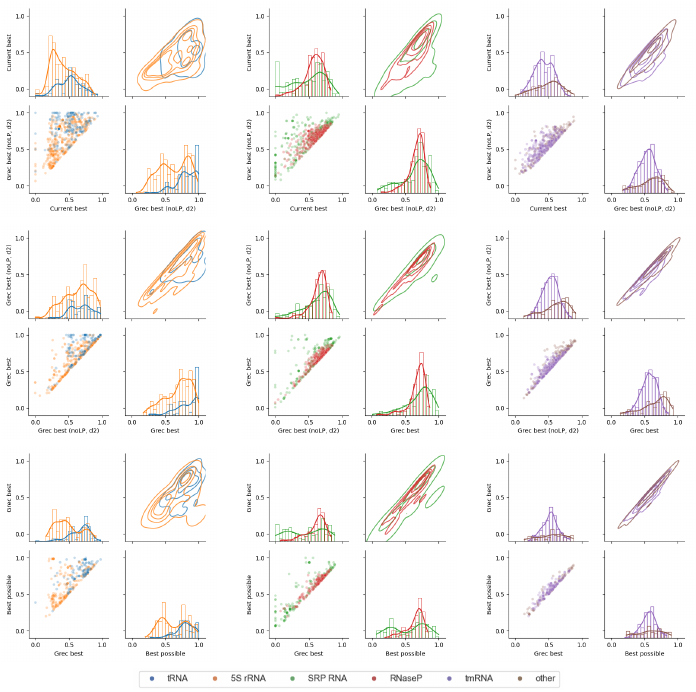}
\caption{Positive accuracy improvements as the set of branching structures is expanded. Grec best with default options offers an improvement over the current standard predictions -- MFE and centroid -- for 47\% of the sequences (top row). Considering Grec best structures when some options are changed offers additional improvement for 40\% of the sequences. Grec best is taken as the best over all 4 possible options (noLP/yesLP, d1/d2) (middle row). Additional gains in accuracy are observed for 25\% of the sequences when the search rectangle is expanded (bottom panel). The histograms are scaled so that 1.0 represents 100 sequences}\label{fig:nesting}
\end{figure}

Of the 1843 sequences for which Grec best (noLP, d2) is an improvement over both the MFE and centroid structures, toggling the options when partitioning Grec yields an improvement for 1592 sequences. Furthermore, we computed the best accuracy in the positive unskewed plane (so that the minimum penalty $a+3c > 0$), and we found that additional gains in accuracy can be gained for 1005 sequences. These results, presented in Figure~\ref{fig:nesting}, indicate that toggling some of the default options, and even expanding the parameter search space should be considered, when alternative branching structures are generated.

\section{Discussion}~\label{sec:discussion}

As we have demonstrated in Section~\ref{sec:results}, generating alternative configurations under branching reparametrization yields a valuable new set of structures which could improve the prediction accuracy. With the algorithms presented in Section~\ref{sec:algs}, this can now be done efficiently for longer sequences. But an important question remains: \emph{How do we find the needle in the haystack, i.e., how do we generate an accurate prediction?}

In preliminary considerations, we considered the ensembles of different branching structures for the RnaseP and SRP RNA sequences with median GC content: \textit{RNase P P.alcalifaciens} (301 nt, pairing density = 0.55, GC content = 0.58) and \textit{SRP Gibb.zeae.\_GSP-229533} (298nt, pairing density = 0.54,
GC content = 0.59). Their plane partitions are shown in~Figure~\ref{fig:example_partitions} (these were generated under the default command line options -- yesLP, d2). 

\begin{figure}
\centering

\includegraphics[width=1\textwidth]{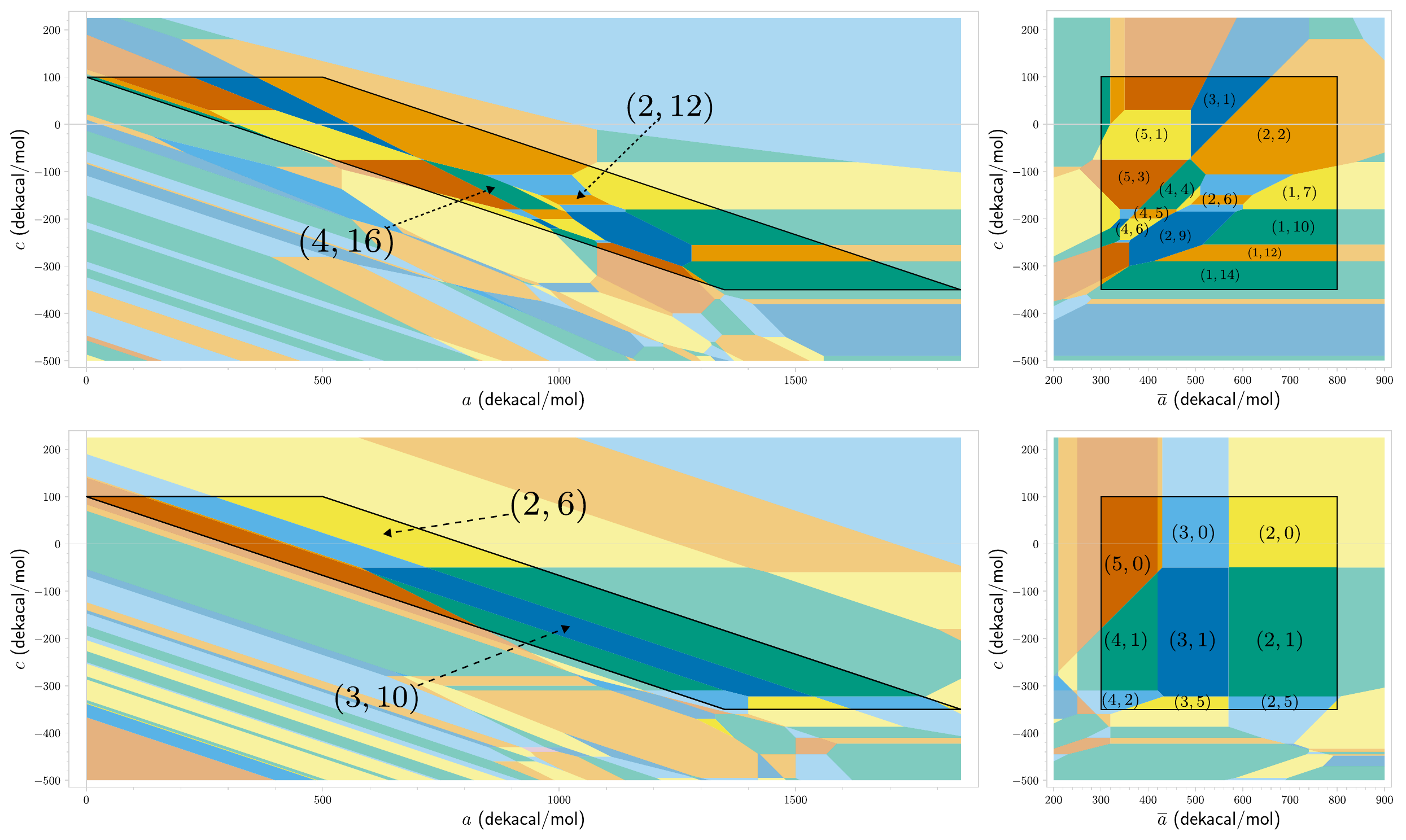}
\caption{The partitions of \textit{RNase P P.alcalifaciens} (top) and \textit{SRP Gibb.zeae.\_GSP-229533} (bottom) restricted to Grec. In the left panels, the regions are shown in the classical $(a,c)$ coordinates and are labeled by the $(x,z)$ components of the corresponding branching signature. On the right,  the regions are shown in the transformed $(\bar{a},c)$-plane ($\bar{a} = a + 3c$) and are labeled by the transformed signature components $(x, \bar{z})$, ($\bar{z}=z-3x$ is the excess branching)}\label{fig:example_partitions}
\end{figure}

We observed some challenges for finding/evaluating the best prediction, which we discuss below. We use $(i,j,k)$ to denote a stack of $k$ consecutive basepairs between $(i,j)$ and $(i+k-1,j-k+1)$. We refer to the structures by the transformed signature components $(x,\bar{z})$, as labeled in the right panels in Figure~\ref{fig:example_partitions}.

\textbf{RNase P}. The target structure has a pseudoknot $(22, 194,4)$ which cannot be predicted with NNTM optimization. Excluding these pairings, the target has size 83 and MFE prediction accuracy is 0.49 (Figure~\ref{fig:rnasep-struct}, top left).

\begin{figure}
\centering
\includegraphics[width=.7\linewidth]{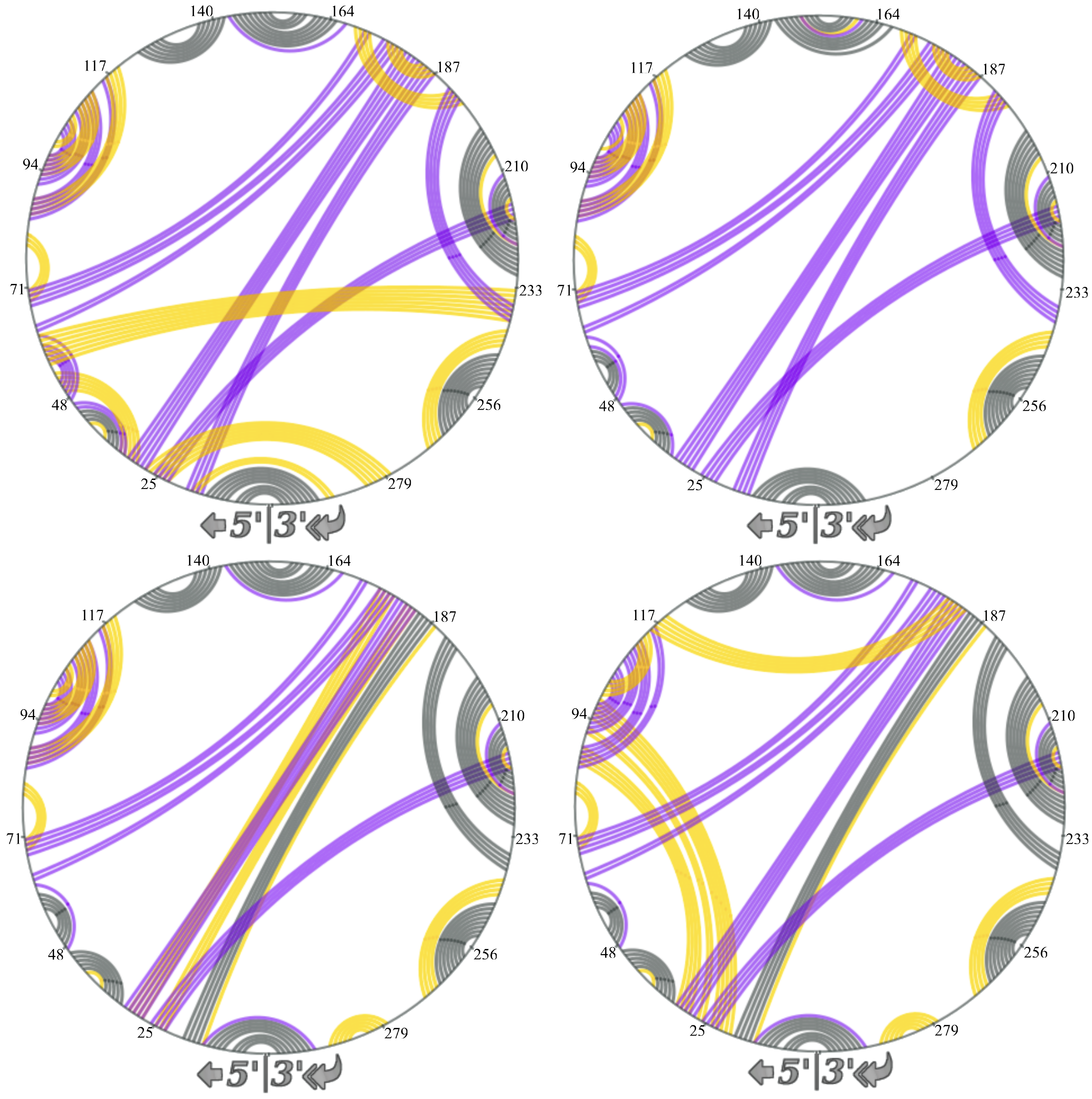}
\caption{\textit{Arc diagrams for RNase P P.alcalifaciens}. The target structure is shown alongside the MFE prediction (top left), the Boltzmann centroid (top right), and two alternative structures generated by Grec: structure $(2,6)$ (bottom left) and structure $(4,4)$ (bottom right). Shared pairs are colored gray, the target basepairs are purple and the prediction is gold.}\label{fig:rnasep-struct}
\end{figure}

The centroid prediction achieves an improved F1 accuracy of 0.58, attributed to an increase in true positive (TP) base pairs from 43 to 47 and a decrease in false positives (FP) from 49 to 29. While it accurately captures local base pairings, it fails to predict any long-range interactions, including $(190, 239, 4)$ (Figure~\ref{fig:rnasep-struct}, top right).

Generating conformations from Grec yields 18 alternative structures, with sizes ranging from 88 to 95 base pairs. Of these, 13 structures exhibit higher F1 accuracy than the MFE prediction, either by increasing the number of TPs or reducing FPs.

A closer examination reveals recurring patterns among these structures. Notably, four predictions from Grec contain 56 TPs—the highest observed in the ensemble. Three of these (structures $(4,4)$,$(4,5)$, and $(4,6)$) share 91 total base pairs, while the fourth one, structure $(2,6)$, shares 72 base pairs with them, including all 56 TPs.

Highest Grec accuracy is obtained by structure $(2,6)$ -- accuracy 0.647 -- but structure $(4,4)$ is essentially the same in F1 score  -- accuracy 0.640.

Despite these similarities, key differences emerge in their FP predictions (Figure~\ref{fig:rnasep-struct}, bottom). For instance, the target includes a 12-base-pair stem with three bulges in the region $[84,116]$. Structure $(2,6)$ predicts a similar stem in $[85,119]$, comprising 11 base pairs and two internal loops. Although no base pairs are exactly shared, this example highlights a limitation of traditional base-pair-level accuracy metrics in recognizing structures that are \emph{close} to the target. Similar discrepancies are seen in helices such as $(22,182,3)$ and $(27,177,4)$, which approximate the target helix $(27,182,5)$, but these similarities are not captured by direct base pair comparison. 

Both structure $(2,6)$ and structure $(4,4)$ correctly predict the helix (15,186,4) as (14,187,5). The remaining false negatives in these structures primarily consist of sparse, long-range interactions -- e.g., $(62,177,2)$, $(67,175,2)$, $(69,171,2)$ -- which are likely tertiary contacts and remain challenging to predict. Nevertheless, the overall prediction quality of structure $(4,4)$ is considerably lower than that of structure $(2,6)$.

This suggests that accuracy evaluations might benefit from considering a structural edit distance rather than strict pairwise matching. Simply expanding the definition of a ``true positive'' to include near matches~\cite{mathews-19}—such as $(i-1,j)$, $(i+1,j)$, $(i,j-1)$, or $(i,j+1)$—does not resolve cases like the region $[85,119]$. Additionally, considering a modified version of positive predictive value where predicted base-pairs which are neither true positives, contradicting or inconsistent with the target structure are labeled compatible and are considered neutral with respect to algorithm accuracy, is one way to adjust for the fact that NNTM could overpredict basepairs without significantly changing the structure~\cite{gardner2004comprehensive}.

\textbf{SRP RNA}. The target region spanning $[21, 46] \times [263, 290]$ is very sparsely paired with only 3 base pairs, which NNTM optimization can not predict accurately (Figure~\ref{fig:srp-struct}). It is unclear whether single-stranded bases in this region should be interpreted as a true signal of unpaired status or as evidence of an incomplete structure.

\begin{figure}
\centering
\includegraphics[width=.7\linewidth]{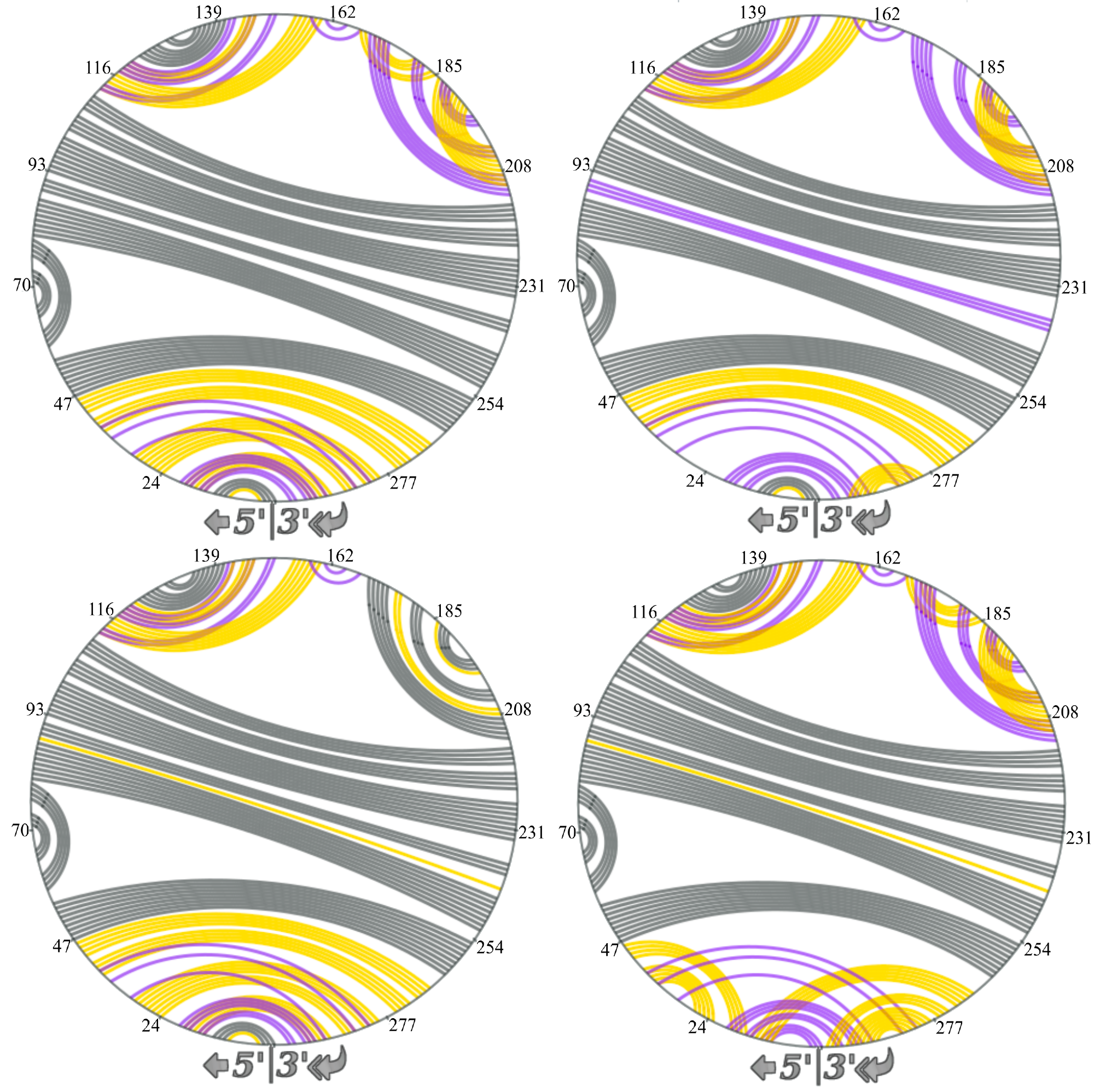}
\caption{\textit{Arc diagrams for SRP Gibb-zeae-GSP-229533m}. The target structure against the MFE prediction (top left), the Boltzmann centroid (top right) and two alternative structures from Grec: structure $(2,0)$ (bottom left) and structure $(3,1)$ (bottom right). Shared pairs are colored gray, the target basepairs are purple and the prediction is gold.}\label{fig:srp-struct}
\end{figure}

In contrast, the target structure contains 43 long-range base pairs densely distributed in the region $[47, 110] \times [214, 262]$. These are more consistently predicted: five Grec-generated structures (structures $(2,0)$, $(2,1)$, $(3,0)$, $(3,1)$, and $(4,0)$) capture them as exact helices, and three additional structures ($(4,1)$, $(4,2)$, and $(5,0)$) produce close approximations.

The F1 scores of the MFE and centroid structures are $0.60$ and $0.61$, repsectively (Figure~\ref{fig:srp-struct}, top). Among the 11 structures generated by Grec, four outperform the MFE in terms of accuracy. Notably, the predicted structures are highly modular, with considerable similarity among them—only structures $(1,5)$,  $(2,5)$, and  $(3,5)$ show significant divergence.

Structure $(2,0)$ achieves the highest accuracy, with an F1 score of 0.72. In contrast, structure $(3,1)$ has an F1 of only 0.56. However, without knowledge of the true structure, it is challenging to determine which of these alternative predictions is more reliable (Figure~\ref{fig:srp-struct}, bottom). Structure $(2,1)$ closely resembles structure $(3,1)$ in the $[167, 210]$ region but matches structure $(2,0)$ elsewhere. Structure $(3,0)$ is similar to structure $(3,1)$ except in $[168, 210]$, where it aligns with structure $(2,0)$. Structure $(2,0)$ and structure $(3,0)$ differ primarily in the $[1, 46] \times [263, 298]$ region.

These variations underscore the difficulty of pinpointing a single most plausible biological structure, while also offering a set of alternative conformations that may be valuable for developing experimentally testable hypotheses -- particularly if the structural differences are biologically meaningful or functionally relevant.

\section{Conclusions}
This work presents an efficient algorithm for generating alternative RNA secondary structures by reparameterizing the branching entropy function in the nearest-neighbor thermodynamic model (NNTM). The method combines minimum free energy (MFE) calculations with half-plane intersection techniques to partition the parameter space into regions, each corresponding to a distinct optimal structure. The algorithm operates in time linear to the number of nodes and partition regions, making it an optimal approach for recovering the full parameter partition while minimizing the number of costly MFE computations that dominate the overall runtime.

We applied our method to assess potential improvements in RNA secondary structure prediction using conformations from Grec compared to standard NNTM predictions, including the MFE and Boltzmann centroid structures. Results show that significant improvement is possible for all 10 families in the benchmarking dataset, indicating that this alternative structure set provides valuable candidates beyond those offered by standard approaches. Furthermore, improvements increase significantly for a large number of sequences if the options for dangling end model and allowing lonely basepairs are changed, or if the parameter search space is increased.

Despite these advancements, the question of how to produce a prediction from an ensemble of different branching configurations remains. The analysis of the ensemble of two sequences highlights the limitations of the NNTM to predict pseudoknots, sparsely base-paired regions, and sparse long-range pairings. Furthermore, it highlights the limitations of conventional accuracy metrics — such as base pair matching — in capturing biologically plausible but imperfect structural predictions. Some examples show modularity among the predicted structures, with clusters of predictions sharing structural elements, suggesting that the ensemble may be used in conjunction with additional structure data or for development of experimentally testable hypotheses.

The NNTM model used for the calculations incorporates linear penalties for multiloops, a choice made originally for computational efficiency. However, subsequent comparisons of prediction accuracy using nonlinear models—such as those based on logarithmic or polymer-theory-derived penalties—showed that the linear model performs best for generating structures~\cite{ward2017advanced}. For this reason, at the outset of this project, we decided it was worthwhile to explore alternative branching structures within the framework of the linear model. Our analysis of the obtained samples supports this approach. Nevertheless, re-evaluating under more sophisticated energy
models~\cite{liu2011fluorescence, zuber2024estimating} has been successfully applied in
other prediction methods that rely on large samples, and they could potentially help identify secondary structures in the samples produced by our algorithm.

The sample of structures generated by the algorithm depends on fixing the penalty $b$ for unpaired nucleotides in the branching loops. In the last two standard combinations of triples, T99~\cite{mathews1999expanded} and T04~\cite{mathews-etal-04}, this parameter was set to 0. In our previous work, which examined parameters for better predictions on a test set of tRNA and 5S rRNA sequences, we were led to investigate several combinations~\cite{polystats, bnb}. Across all the new parameter combinations considered, we found that the $b$ values clustered within a small range, roughly centered around 0. As a result, we concluded that it is reasonable to specialize the $b$ value and focus on the different branching structures based on the $a$ and $c$ trade-offs. However, we note that (1) specializing to a different $b$ value is possible and (2) the algorithm can be adapted to specialize a value of $a$ or a value of $c$ if desired. Expanding the search range in these ways could potentially improve accuracy, and exploring this would be interesting when an appropriate filtering of structures in the sample is developed.

Finally, we note that the mathematical background of our work is related to the work on the problem of identifying regions where certain optimal solutions are valid -- an important problem in mathematical programming, since in many applications, parameters like coefficients in the objective function vary over time or due to uncertainty (see e.g.~\cite{gal1972multiparametric, mavnas1968finding, matheiss1980survey, borrelli2003geometric}). So, while our focus here was developing algorithms for exploring RNA branching configurations which are compatible with existing MFE optimization tools, our algorithms might be useful in other settings as well. Namely, any linear function $f$ on a finite set of vectors with rational coordinates, induces a subdivision of the parameter space into convex sets which correspond to the different minima (or, alternatively, maxima) of $f$. Then one could use the algorithms we presented to recover that partition, provided that one has a consistent oracle which produces a minimum (respectively, a maximum) for $f$ for any choice of parameters from $\mathbb{Q}$ or $\mathbb{Z}$.

\newpage

\input{revision1.bbl}


\end{document}